\documentclass[]{rsos}

\usepackage{graphics}
\usepackage[]{graphicx}

\DeclareGraphicsExtensions{.jpg,.pdf,.tif,.png,.tiff,.eps}

\usepackage{hyperref}
\usepackage{url}
\usepackage{amssymb, amsmath}
\usepackage{mathtools}%
 
\usepackage{enumerate}%

\newtheorem{thm}{Theorem}
\newtheorem{lem}[thm]{Lemma}%
\newtheorem{coro}[thm]{Corollary}%
\newtheorem{asm}[thm]{Assumption}%
\numberwithin{equation}{section}%

\renewcommand{\thefigure}{\arabic{figure}}%
\renewcommand{\thetable}{\arabic{table}}%

\newcommand{\norm}[1]{\left\Vert#1\right\Vert}
\newcommand{\snorm}[1]{\left\Vert#1\right\Vert_\infty}
\newcommand{\inn}[1]{\left\langle#1\right\rangle}
\newcommand{\abs}[1]{\left\vert#1\right\vert}
\newcommand{\set}[1]{ \left\{ #1 \right\} }
\newcommand{\innp}[1]{\left( #1 \right)}

\renewcommand{\b}[1]{\mathbb{#1}}
\newcommand{\s}[1]{\mathcal{#1}}

\renewcommand{\d}[1]{\mathrm{d}#1}

\renewcommand{\l}{\ell}

  \newcommand{\nn}{\nonumber\\}
\newcommand{\todistribution}{\xrightarrow{d}}
\newcommand{\toprobability}{\xrightarrow{p}}
\newcommand{\toas}{\xrightarrow{a.s.}}
  
  \newcommand{\dd}{\mathrm{d}}
\usepackage{listings}

\begin{document}

\title{Asymptotic convergence in distribution of the area bounded by prevalence-weighted Kaplan-Meier curves using empirical process modeling}
\author{Aaron Heuser$^2$, Minh Huynh$^2$, Joshua C. Chang$^1$}
\address{$^1$ Epidemiology and Biostatistics Section,  Rehabilitation Medicine Department, 
The National Institutes of Health, Clinical Center, Bethesda, Maryland 20892, U.S.A. \\
$^2$ Impaq International LLC, Washington, DC 20005, U.S.A. \\
}

\corres{Joshua C. Chang \\
\email{josh.chang@nih.gov}} 

  \begin{abstract}
 
The Kaplan-Meier product-limit estimator is a simple and powerful tool in time to event analysis. An extension exists 
for populations stratified into cohorts where a population survival curve is generated
by weighted averaging of cohort-level survival curves. 
 For making population-level comparisons using this statistic, we 
analyze the statistics of the area between two such weighted survival curves.
 We derive the large sample behavior of this statistic based
on an empirical process of product-limit estimators. This estimator
was used by an interdisciplinary NIH-SSA team in the identification of medical
conditions to prioritize for adjudication in disability benefits processing.

  \end{abstract}

\begin{fmtext}

\end{fmtext}

\maketitle

  \section{Introduction}
    \label{s:intro}

Survival analysis addresses the classical statistical problem of determining characteristics
of the waiting time until an event, canonically death, from observations of 
their occurrence sampled from within a population. This problem is not trivial as the expected waiting
time is typically dependent on the time-already-waited. 
 For instance, a hundred-year-old can be more certain
 of surviving to his or her one hundred and-first birthday than a newborn might reasonably be.
 However, the comparison may shift in the newborn's favor for the living to one-hundred and twenty-one,
 particularly in light of medical advances that make survival probabilities non-stationary.
Parametric approaches for assembling survival curves
are usually not flexible enough to capture this complexity.

One simple approach to this problem was pioneered by the work of Kaplan and Meier \cite{kaplan1958nonparametric}.
Their product-limit estimator 
    \cite{gill1980censoring,gill1981large,van1996new,shorack2009empirical} is a non-parametric statistic that is used for
  inferring the survival function for members of a population from observed lifetimes.
This method is particularly useful in that it naturally handles the presence of right censoring,
 where some event-times are only partially observed because they fall outside the observation window.
It was not, however, designed to account for varying subpopulations that may yield non-homogeneity
in overall population survival (Fig.~\ref{fig:fig1}). For instance, in the example given above, subpopulations for survival
characteristics may be defined by birth year or entry cohort of a subject in a particular study (Fig.~\ref{fig:fig1}).

 Several existing statistical methods address variants of this limitation.
A natural approach is to consider the varying subpopulations as defining underlying 
    covariates, thus laying the framework for a proportional hazards model. The assumption
    of proportional hazards is quite strong. When considering time-dependent 
    statistics (as in the motivational example), it is violated in all but a few specific cases.
    Likewise, frailty models, first developed by Hougaard (cf.~\cite{hougaard1984life}), and 
    extended by Aalen (cf.~\cite{aalen1994effects}), assume multivariate event 
    distributions, but also make assumptions on the underlying event 
    distributions and assume proportional hazards. 
    
    Other existing methods, such as bivariate survival analysis 
    (cf.~\cite{lin1993simple}), consider the time to 
    observation and the time to event as conditionally independent random times. 
    Underlying these methods is the assumption that upon the time of observation, 
    all individuals will then have a similar event time distribution, thus 
    failing to acknowledge the temporal changes.

These complexities arose in the identification of new disorders to incorporate into the United States Social Security
Administration (SSA)'s Compassionate Allowances (CAL) initiative. The CAL initiative seeks to identify candidate
medical conditions for fast-tracking in the processing of disability applications. The intent of this initiative is 
to prioritize applicants who are most likely to die in the time-course of usual case processing so that they
may receive benefits while still living. 

At its inception, the CAL initiative identified conditions based on the counsel of expert opinion~\cite{rasch2014first}.
The SSA in collaboration with the National Institutes of Health (NIH) sought to expand the list of CAL
conditions systematically, using a data-based approach.
Using in-part the survival estimator described in this manuscript, the NIH identified 24 conditions for 
inclusion into the list of conditions~\cite{rasch2014first}.

\begin{figure}
\center
\includegraphics[width=21pc]{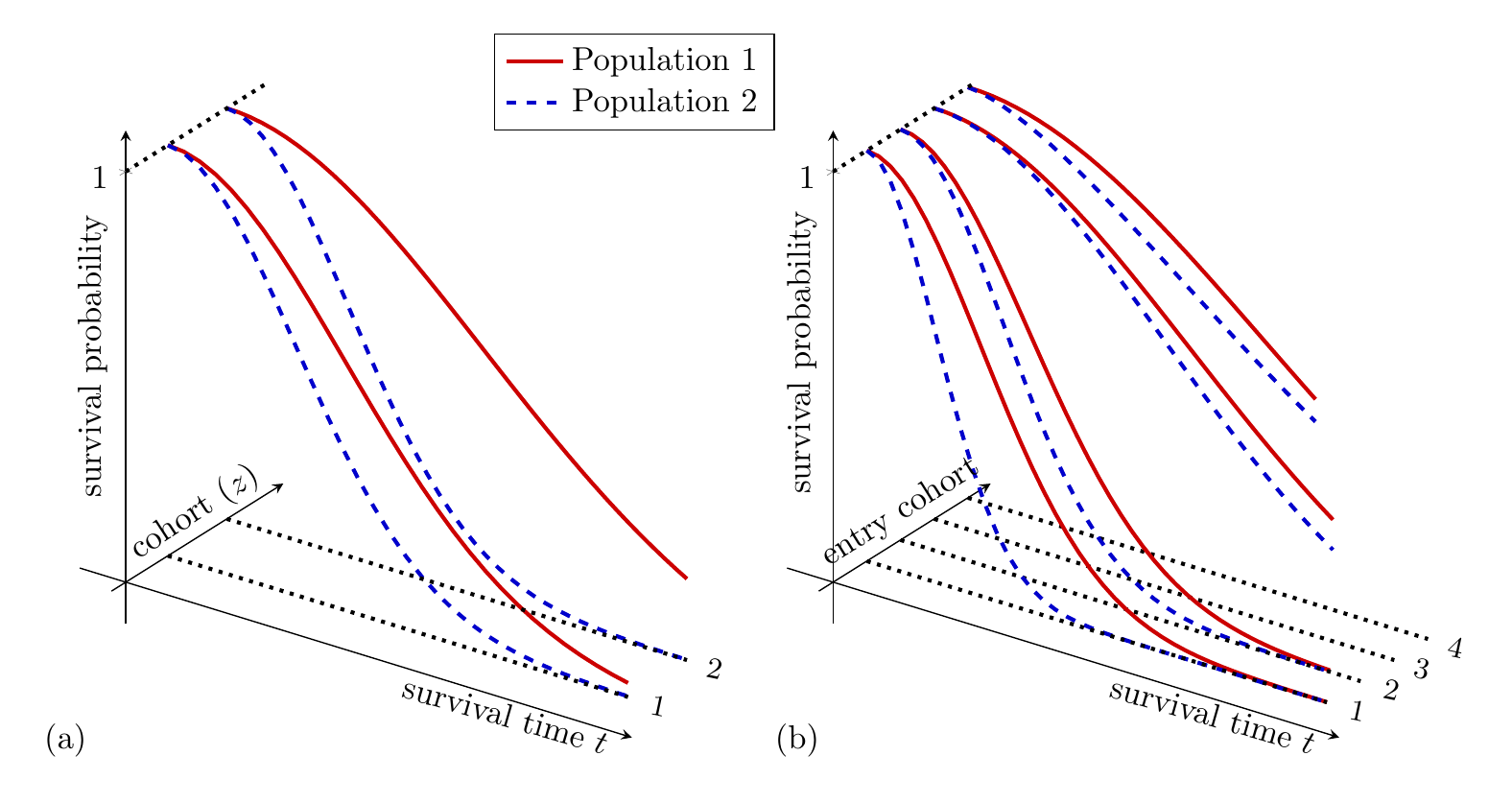}
\caption{\textbf{Inhomogeneity} of survival within populations can result due to at least two reasons. In (a), inhomogeneity results from a categorical covariate that influences survival statistics. In (b), inhomogeneity results from non-stationarity, where cohorts of individuals are sampled at different times. In this case, the problem of progressive censoring is apparent because later cohorts have not been observed as long.} 
\label{fig:fig1}
\end{figure}



    The methodology used in CAL is related to that of the work of Pepe and Fleming (cf.~\cite{pepe1989weighted,pepe1991weighted}), 
    where a class of weighted Kaplan-Meier statistics is introduced. Though these statistics
    exhibit the same limitations as in the standard Kaplan-Meier case, it 
    should be noted that \cite{pepe1991weighted} introduces the stratified weighted 
    Kaplan-Meier statistic. The statistic presented here is a priori quite 
    similar, but instead of a weighting function, includes the empirical 
    prevalence. In doing so, the weight is no longer independent of the event 
    time estimate, and thus requires much different methods of proof.
     
    We thus consider the overall survival distribution for a 
    population of individuals with sub-populations that exhibit 
    non-homogeneous survival distributions. Through this consideration, a 
    new test statistic, based upon the empirical process of 
    product-limit estimators is developed. Through constructive methods, 
    this test-statistic compares survival distributions among the distinct 
    subpopulations, and weights according to distribution of the 
    identified subgroups.

\section{Statistical method}


Suppose $\Gamma^{(1)}$ and $\Gamma^{(2)}$ are disjoint populations of individuals where
each individual belongs to exactly one of $d$ distinct cohorts labeled $z\in\mathbb{Z}_d$.
For randomly selected individuals $\gamma  \in \Gamma^{(i)}$ within population $i$, we desire to understand the 
statistics of the event time $T^\gamma$  under the assumption that survival is
 conditional on cohort $z^\gamma$ and population.

One representation of the marginal survival probability for members of population $i$, 
$\theta^{(i)}_t = \b{P}{\set{T^\gamma > t \mid\gamma\in\Gamma^{(i)}}},$ is found by conditioning
on cohort,
    \begin{equation}
      \theta^{(i)}_t =  \sum_{z=1}^d \underbrace{\b{P}{\set{T^\gamma>t\mid z^\gamma=z, \gamma\in\Gamma^{(i)}}}}_{S^{(i)}_{z,t}} 
      \underbrace{\b{P}{\set{z^\gamma=z\mid \gamma\in\Gamma^{(i)}}}}_{q^{(i)}_z}, 
      \label{eq:conditioningsum}
    \end{equation}
    where $S^{(i)}_{z,t}$ represents the survival function for individuals of cohort $z$ in population $i$,
    where each individual's cohort membership is known. 
  
    We use this representation of the survival probability as motivation to formulate an 
    estimator for the population-average survival functions
    \begin{equation}
    \hat{\theta}^{(i)}_t = \sum_{z=1}^d \hat{q}^{(i)}_z \hat{S}^{(i)}_{z,t},\label{eq:thetat}
    \end{equation}
    where $\hat{q}^{(i)}_z$ and $\hat{S}^{(i)}_{z,t}$ are estimators of the cohort prevalence and 
    cohort-wise survival respectively.  This weighted Kaplan Meier method has appeared previously in the literature~\cite{murray2001using},
    and has been empirically validated against the pure Kaplan Meier method~\cite{zare2014comparison}, where the weighting
    procedure was found to reduce the bias in the construction of survival curves.
    The asymptotic convergence of 
    the product-limit estimator and weighted variants is well established~\cite{cai1998asymptotic, pepe1991weighted}.
     We use this survival curve reconstruction method as a base in constructing
    a new statistic for comparing populations. The focus of this manuscript is not the properties of this 
    survival estimator but rather the asymptotic convergence of its bounding area and the use of such
    a quantity for evaluating a null hypothesis.
    
    Our concern is the general situation where random samples of size $n^{(i)}$ are chosen from each 
    of the respective populations. Within these samples, the number of individuals within each cohort, $n_z^{(i)},$ is 
    counted, from which an estimator of the cohort distribution is obtained,
    \begin{equation}
    \hat{q}_z^{(i)} = \frac{n^{(i)}_z}{n^{(i)}}. \label{eq:q_z}
    \end{equation}
    In turn, we assume that the cohort-level survival functions $\hat{S}^{(i)}_{z,t}$ are estimated independently
     using  the product-limit estimator. Note that since the product limit estimator is not a linear functional
     of sampled lifetimes, $\hat{\theta}_t^{(i)}$ is distinct from the estimator 
     obtained by applying the product limit estimator directly on all $n^{(i)}$ samples of population $i$. To
     prevent confusion, we denote all direct applications of the  product-limit estimator using $\hat{S}$ and all 
     instances of weighted sums of product limit estimators using the Greek letter $\hat\theta.$ 
     
     With these elements in place, we define our test statistic
\begin{equation}
\hat{\Theta} = \sqrt{\frac{n^{(1)}n^{(2)}}{n^{(1)}+n^{(2)}}} \int_0^\tau \dd{t}  \left(\hat{\theta}^{(1)}_{t} -\hat{\theta}^{(2)}_{t} \right),\label{eq:Theta}
\end{equation}
where $\tau=\inf\set{\tau_z: z\in\mathbb{Z}_d}$, and $\tau_z$ denotes the time at which cohort $z$ is censored in observations. Note that in the absence of censoring this statistic is equivalent to comparison of mean lifetimes between the two populations~\cite{pepe1989weighted}. We state here the 
main result of the paper -- the large sample behavior of this statistic within a null-hypothesis statistical testing framework.
   \begin{thm}
      \label{thm:main}
      Let $C^{(i)}_{z, t}$ denote the probability that 
      a $z$-type individual has not yet been censored at time $t \geq 0$ (the survival probability
      relative to the occurrence of censoring), and 
      $q^{(i)}_z$ denote the probability that an individual in population $i$ is of cohort 
      $z$, and let $p^{(i)} = n^{(i)}/(n^{(1)}+n^{(2)}).$ Suppose that $\theta^{(1)}_t = \theta^{(2)}_t.$ Then $\hat{\Theta}\xrightarrow{d} N(0,\sigma^{2})$, as $n^{(i)} \to \infty$, with
      \begin{align*}
        \sigma^{2} &= \sum_{i = 1}^{2} (1-p^{(i)})\innp{\sum_{z = 1}^d q^{(i)}_z \phi^{2}_z 
        - \innp{\sum_{z = 1}^d q^{(i)}_z \phi_z}^{2}} \nn
          &\quad- \sum_{z = 1}^d \int_0^{\tau_z}\!\d{S}_{z, t}\,
          W_{z, t}\times\innp{\frac{\phi_{z, t}}{S_{z, t}}}^{2},
      \end{align*}
      where for $0 \leq t \wedge \tau_z$, where $\tau_z$ is the time at which samples of cohort $z$ are censored,
       $\phi_{z, 
      t} = \int_t^{\tau_z}\!\d{s}\, S_{z, s}$, $\phi_z \equiv \phi_{z, 
      0}$, $S_{z,t}$ is the survival function for the pooled data of cohort $z$, and
      \begin{align*}
        W_{z, t} =  \innp{\frac{p^{(1)} C^{(1)}_{z, t-} q^{(2)}_z + p^{(2)} C^{(2)}_{z, t-} 
        q^{(1)}_z}{C^{(1)}_{z, t-} C^{(2)}_{z, t-}}}.
      \end{align*}
   Note that this quantity is well-defined since by definition of $\tau_z$, ${C}^{(z)}_{z,t} >0 $ for all $t\leq \tau_z$.   The variance $\sigma^{2}$ may be consistently estimated by
      \begin{align}
        \hat{\sigma}^{2} &= \sum_{i = 1}^{2} (1-p^{(i)})\innp{\sum_{z = 1}^d 
        \hat{q}^{(i)}_z \hat{\phi}_z^{2} - \innp{\sum_{z = 1}^d \hat{q}^{(i)}_z 
      \hat{\phi}_z}^{2}} \nn
        &\quad- \sum_{z = 1}^d \int_0^{\tau_z}\!\d{\hat{S}}_{z, t}\,
        \hat{W}_{z, t}\times\innp{\frac{\hat{\phi}_{z, t}}{\hat{S}_{z, t}}}^{2},
                \label{eqn:sigma2}
      \end{align}
      where for  $0 \leq t \wedge \tau_z$, 
      $\hat{S}_{z, t}$ is the product-limit estimate of the pooled data for cohort $z$, 
            \begin{equation}
      \hat{\phi}_{z, t} = \int_t^{\tau_z}\d{s}\, \hat{S}_{z, s},
      \end{equation}
      $\hat{C}^{(i)}_{z, t}$ is the product-limit estimate associated to the event 
      of censoring for cohort $z$ within population $i$, 
       $ \hat{\phi}_z \equiv \hat{\phi}_{z, 
      0},$ and
      \begin{align}
        \hat{W}_{z, t} =  \innp{\frac{p^{(1)} \hat{C}^{(1)}_{z, t-} \hat{q}^{(2)}_z 
          + p^{(2)} \hat{C}^{(2)}_{z, t-} \hat{q}^{(1)}_z}{\hat{C}^{(1)}_{z, t-} 
        \hat{C}^{(2)}_{z, t-}}}. \label{eq:What}
      \end{align}
    \end{thm}
    Note that this quantity is also well-defined since $\hat{C}^{(z)}_{z,t} >0 $ for all $t\leq \tau_z$.
  In Appendix~\ref{s:mainproof}, we provide a proof of Theorem~\ref{thm:main} in an empirical process framework.  Note that since survival estimates $\hat\theta$ and $\hat{S}$ are step functions, all integrals are exactly
  computable.

  \begin{section}{Numerical investigation}
    \label{s:sim}
    
    A computational implementation of the test statistic $\hat{\Theta}$ and weighted survival estimators is available
    in the form of a package for R. This package also contains a class to handle arithmetic involving right-continuous piecewise linear  functions.   In the appendices we have provided source code that may be used 
    for installing and invoking this package. 
    
    Here, we present a computational investigation of the weighted survival
    curve estimator and the corresponding test statistic. Using simulations, we investigated the statistical power of $\hat{\Theta}$, contrasted with that of existing non-parametric methods. Using a real dataset, we demonstrate
    the computation of $\hat{\Theta}$, $\hat{\theta}_t$, and evaluate Type-I error.
    
    \subsection{Evaluating statistical power through simulations}
    
    Using simulations, we explored the statistical power of the test statistic $\hat\Theta$ in a case
    where populations are difficult to distinguish based purely on mean survival time.
    \begin{figure}
    \includegraphics[width=\linewidth]{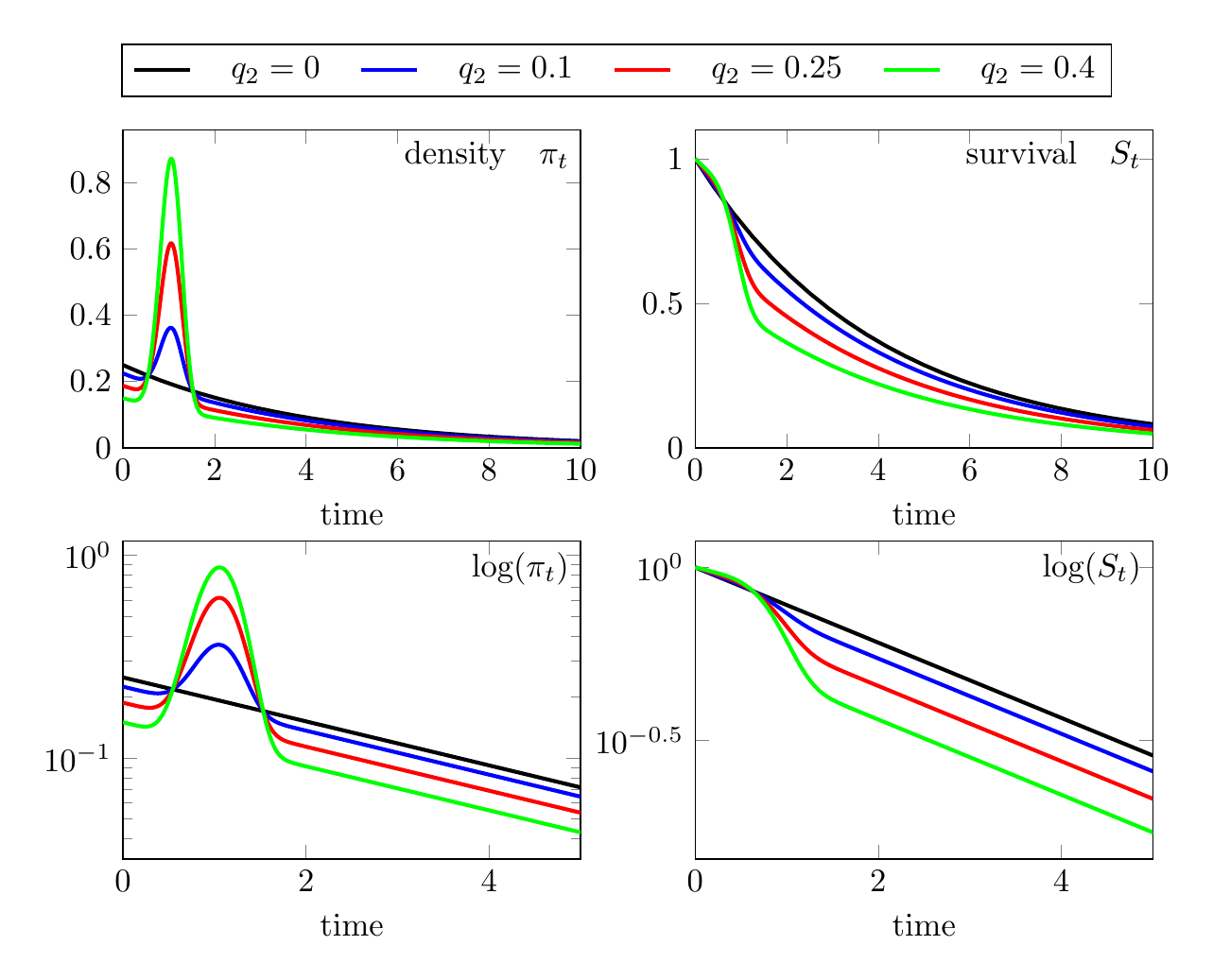}
    \caption{\textbf{Admixture test distributions} used in simulated investigations of our estimator. Populations formed using $q_2\in[0,1)$ admixtures of $(1-q_2)\textrm{exponential}(\lambda=5^{-1})$ and $q_2\textrm{Weibull}(k=5,\lambda=1)$ event time distributions. Event time density functions $\pi_t$ and corresponding survival functions  $S_t$ are shown for various values of $q_2$.
    }
    \label{fig:fig2}
    \end{figure}
    As test populations, we examined admixtures of exponential and Weibull
    distributions for the event time, and compared survival in these mixture populations
    to survival of a population of purely exponential event times (Fig.~\ref{fig:fig2}).
        Population 1 consists of individuals having an 
    exponentially distributed lifetime with a mean of $\lambda^{-1}=4$ years. Population 
    2 consists of two types of individuals: those who have an exponentially 
    distributed lifetime with a mean of $5$ years (type $z=1$), and those of type $z=2$ who have 
    a Weibull distributed lifetime with shape parameter $k=5$ and scale parameter 
    $\lambda=1$.
    
     Since Population $1$ is homogeneous, we only track subpopulations of Population $2$ - we drop the
    superscript and denote the proportion of Population 2's members of type $2$ by $q_2$. 
    It is most instructive to examine our method in the neighborhood where both populations
     have approximately the same expectation value for the event time, which occurs for $q_2\approx 0.245$.
    For this reason, we chose values near $0.25$ for our simulations.
    
        \begin{figure*}
    \includegraphics[width=\linewidth]{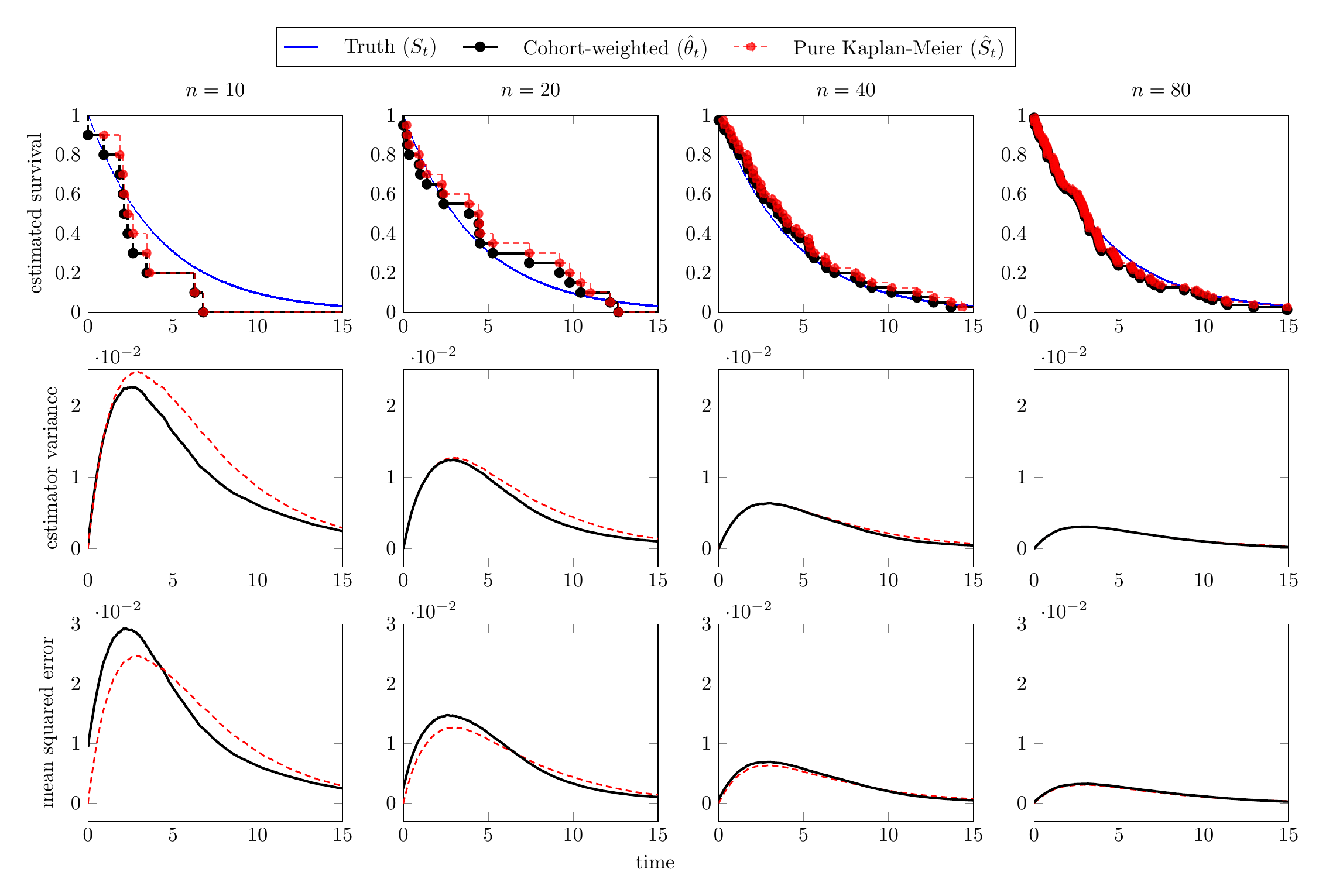}
    \caption{\textbf{Comparing estimators of survival.} The survival estimation method of Eq.~\ref{eq:thetat} compared to pure Kaplan-Meier for a population containing an admixture of $(1-q_2)\mathrm{Exponential}(1/5)$ and $q_2\mathrm{Weibull}(1,5)$ individuals, where $q_2=0.25$. At a given sample size $n$, the survival estimates are obtained (top row: examples shown and contrasted). The estimator variance and  mean square error were approximated using $10,000$ resamplings for each of the sample sizes.}
    \label{fig:fig3}
    \end{figure*}
    
    To compare the reweighted Kaplan-Meier estimator (Eq.~\ref{eq:thetat}) to the standard Kaplan-Meier estimator, we estimated survival for the admixed population for $q_2=0.25,$ using various sample sizes. In Fig.~\ref{fig:fig3}, we present example reconstructions using these two methods. The estimator variance was approximated using $10,000$ resamplings of sample size $n$ of the admixed population, for each value of $n$. The estimation error, as defined by mean-squared difference between the reconstruction and the true survival function, was approximated in the same manner. 
        
    To better-understand the performance of the test statistic (Eq.~\ref{eq:Theta}), we evaluated its statistical power
    against that of other test statistics in distinguishing between Population 1 and Population 2 for various values of $q_2.$ For samples of size $n^{(i)}\in\set{30,50,100,200,1000}$ taken from each population, we performed $1000$ 
    null hypothesis statistical tests using our method, the log-rank method~\cite{berty2010determining}, and 
    the standard Kaplan-Meier Wilcoxon signed-rank difference-of-mean methods~\cite{wilcoxon1945individual,schoenfeld1981asymptotic}. The power of the test, or the proportion of times that the null hypothesis was correctly rejected, is shown in Fig.~\ref{fig:fig4}.

    \begin{figure*}
    \includegraphics[width=\linewidth]{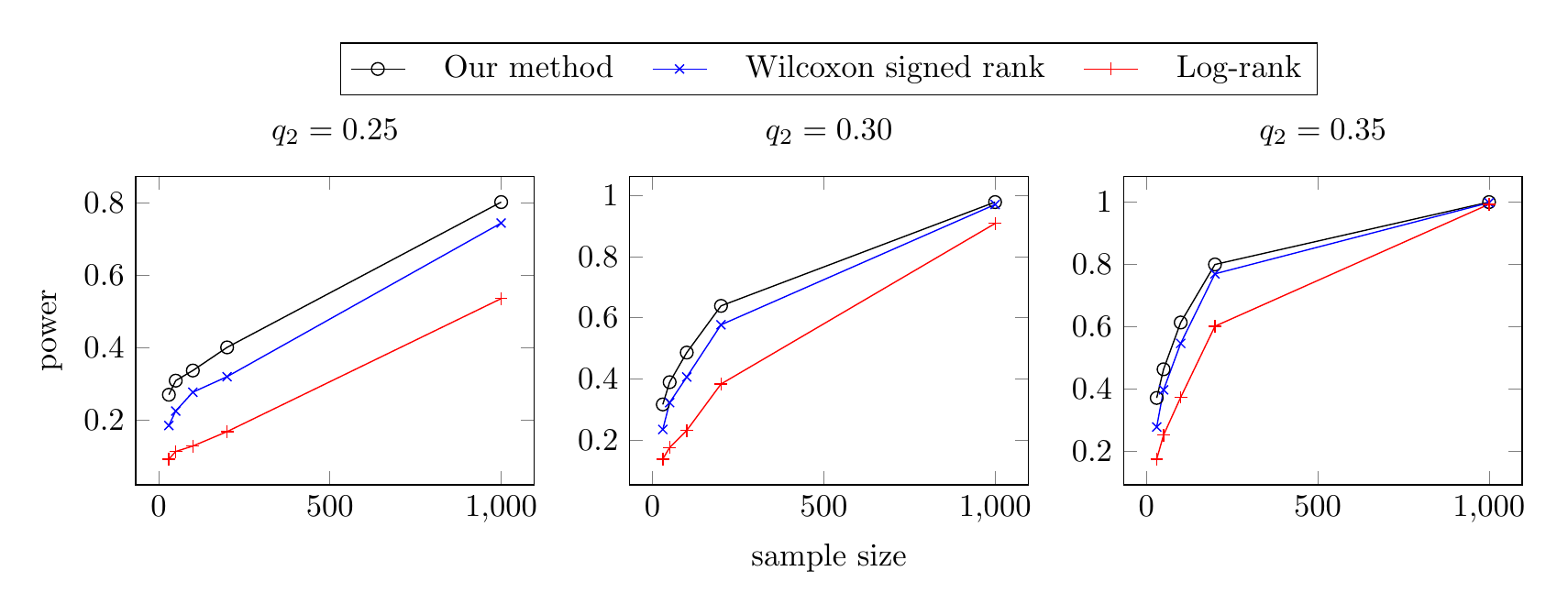}
    \caption{ \textbf{Simulated power computation} comparing exponentially distributed lifetimes against a mixture of $q_2$ Weibull and $(1-q_2)$ exponential distributions, where $q_2$ determines the amount of mixing. A larger value of $q_2$ implies more real difference between the survival functions of the two populations. The power of our method (black) is compared to the power of Kaplan-Meier Wilcoxon signed rank (blue) and Log-rank (red) methods. (More power is better). }
    \label{fig:fig4}
    \end{figure*}

\subsection{Evaluating Type-I error in a real world example}

  \begin{figure*}\center
    \includegraphics[width=0.6\linewidth]{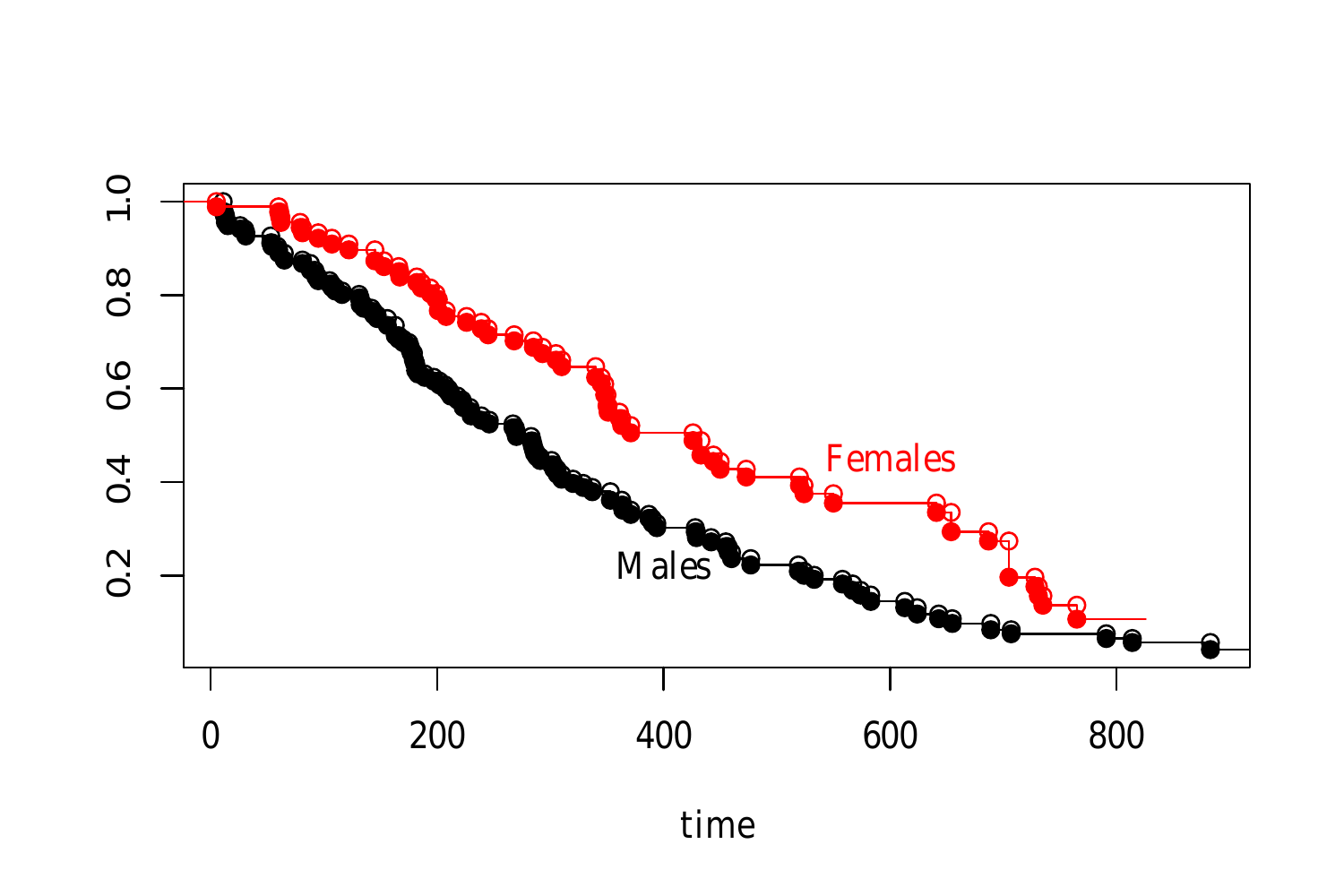}
    \caption{ \textbf{$\hat\theta_t$ estimates} for days of lung cancer survival in males (population 1) versus females (population 2) from the NCCTG lung cancer dataset. The statistic $\hat\Theta$ implies an asymptotic $P$-value of $0.0009$, rejecting $H_0$ at $\alpha=0.05$.}
    \label{fig:fig5}
    \end{figure*}

    \begin{figure*}
    \includegraphics[width=\linewidth]{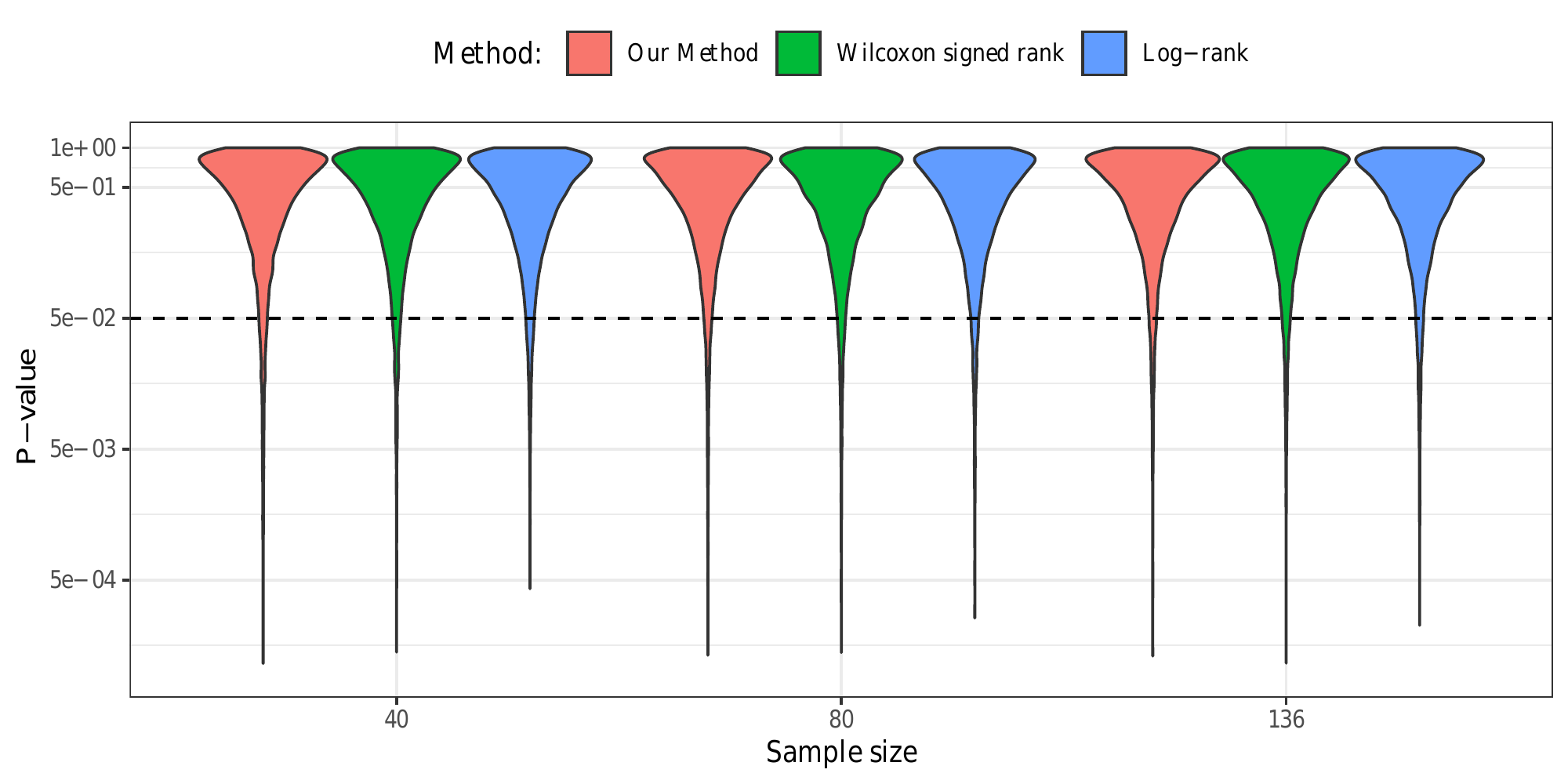}
    \caption{ \textbf{$P$-value distributions} for the comparison between samples of size $n/2$ of two random subpopulations of male patients in the lung cancer data. The proportion of null hypotheses rejected
    at each of the three statistical methods is similar, at approximately $5\%$ for $\alpha=0.05$.}
    \label{fig:fig6}
    \end{figure*}

We applied the survival estimator and statistic to NCCTG Lung Cancer data  \cite{loprinzi1994prospective} available 
within the \texttt{survival} package for R. We compared the survival between male ($n^{(1)}=136$) and 
female ($n^{(2)}=90$) cancer patients, organized by ECOG performance score ($z\in\{0,1,2 \})$ as cohort. 
Using males as population 1 and females as population 2, we arrived at the test-statistic estimate: 
$\hat\Theta = -961$, with 95\% asymptotic confidence interval: $(-1527,-396)$, which
would support rejection ($P\approx 0.0009$) of the null hypothesis ($\hat{\theta}_t^{(1)}=\hat{\theta}_t^{(2)}$) at $\alpha=0.05$. 
For reference, both the Wilcoxon ($P\approx 0.0012$) and log-rank ($P\approx 0.0015$) tests referenced in
Fig.~\ref{fig:fig4} also rejected the null hypothesis.

In theory, the Type-I error is set by the significance level at study design. Whether a statistic controls Type-I error
correctly depends on accurate evaluation of its sampling distribution. In the case of $\hat\Theta$, our main result is that the sampling distribution for this estimator converges asymptotically in distribution to a Gaussian with a definite variance. However,
small-sample behavior is not guaranteed. 
To evaluate Type-I error, we used the same dataset, restricted to male patients. For each of $n\in\{ 40, 80, 136 \}$, we 
sampled without replacement the $n$ male patients split into two groups so that $n^{(1)} = n^{(2)} = n/2$, and
compared survival between the two random groups. Repeating this procedure $10,000$ times, we generated the
observed distribution of $P$-values, presented in Fig.~\ref{fig:fig6} in log-scale. The distributions computed using
the three methods are similar. The three methods all rejected $H_0$ approximately $5\%$ of the time except for 
the case of $\hat{\Theta}$ at $n=40$, which rejected $H_0$ approximately $6\%$ of the time. Essentially, 
asymptotic convergence as defined by the accurate evaluation of $\alpha=0.05$ Type-I error occurs somewhere in between $40$ and $80$ samples for this particular dataset.

Probing deeper, we examined the sampling distributions of $\hat{\Theta}$ for each of $n\in\{50,60,70\}$, in each
 instance compared to the Gaussian distribution stated in Thm.~\ref{thm:main}, where the approximation is computed
 using only the first sample of size $n$. The results for these simulations are shown in Fig.~\ref{fig:fig7}, where it
 is seen that the sampling distribution of $\hat{\Theta}$ is approximately the same as the computed asymptotic Gaussian distribution, which is traced out in red.
 
     \begin{figure*}
    \includegraphics[width=\linewidth]{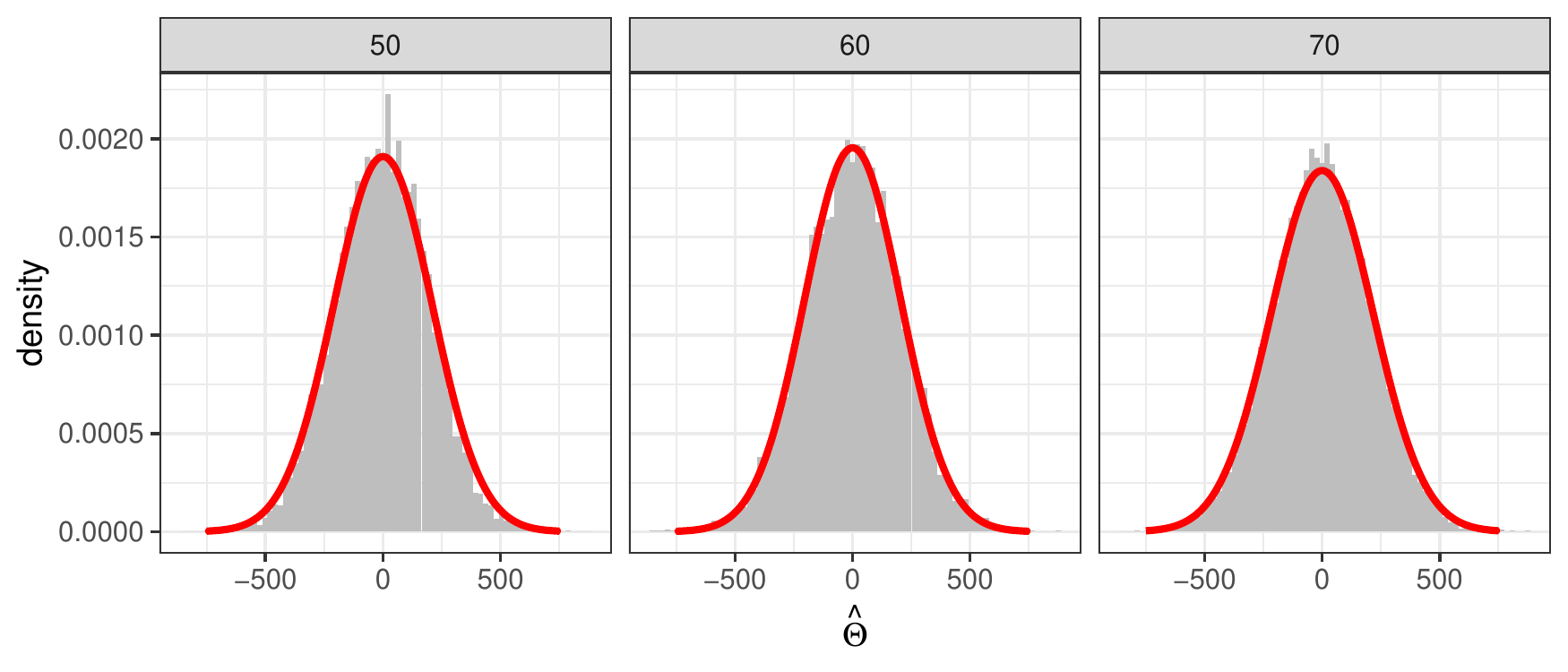}
    \caption{ \textbf{Histograms of $\hat{\Theta}$ sampling distributions} for comparing survival between random subsets of male lung cancer patients using sample sizes of $n\in\{50,60,70\}$. Traced in red, the asymptotic Gaussian density as computed using Thm.~\ref{thm:main} on the first sample set of each size is overlayed.}
    \label{fig:fig7}
    \end{figure*}

The R code used to compute these examples
 is available in Appendix B\ref{sec:realworld}. 

  \end{section}
  
  \section{Discussion and Conclusion}

In this manuscript we have proposed a test statistic that uses a cohort-averaged
survival function estimator in order to make cross-population comparisons
of survival within a null hypothesis statistical testing framework. The proposed survival estimator 
was an empirically-weighted average of cohort-level product-limit estimates. The test statistic
involved computation of the area between estimated survival functions for two populations. By
invoking an empirical stochastic process, we proved asymptotic normality of this test statistic.

Using simulations, we contrasted the weighted survival estimator against the pure Kaplan-Meier 
estimator. It is seen, in Fig.~\ref{fig:fig3}, that the survival curves generated from the two methods are
distinct yet similar. In the second and third rows of Fig.~\ref{fig:fig3}, one sees 
 that this reweighted estimator has comparable performance to
the pure Kaplan-Meier estimator at large sample sizes. Asymptotically, both estimators converge to the 
true survival function, with variance converging to zero. At small sample sizes, there are differences. 
The reweighted estimator has reduced variance at the cost of larger bias, in a time-dependent manner.
It also appears to have smaller variance at the cost of larger error at earlier times. This error
at earlier times is mitigated by decreased error at later times (better reconstruction of tails), however, the estimator variance is lower
at all times. Hence, dependent on costs,  
for small samples, this reweighted estimator may be preferable to the pure Kaplan-Meier estimator.

In simulations of the test statistic derived from the reweighted survival estimator, we saw 
superior performance compared to existing methods. In Fig.~\ref{fig:fig4}, it is seen that in all cases,
 the test statistic $\hat{\Theta}$ was better at distinguishing between the two populations than 
either the Wilcoxon signed-rank test or the log-rank test. The relatively-high statistical power of this
 statistic is due to tighter variation in the test-statistic. In nearly all cases $(>99.5\%)$, the estimator
variance for the tested method was less than that of the other two tests (not shown).

This manuscript derives the asymptotic convergence in distribution of the $\hat\Theta$ statistic. 
Numerically, we demonstrated convergence of the statistic in Figs.~\ref{fig:fig6} and~\ref{fig:fig7}, where we
verified that the asymptotic approximation respects Type-I error at $\alpha=0.05$ and where
we observe good match between the sampling distribution of $\hat\Theta$ and the asymptotic
Gaussian distribution provided by Thm.~\ref{thm:main}.

A variant of this method was used in \citet{rasch2014first}
in order to classify physical disorders based on severity for the sake of prioritization of processing
for disability claims. Since the underlying survival surface is non-stationary, and the fixed observation
windows create progressive censoring, that paper illustrates the utility of this statistical method.
In that manuscript, the cohorts were defined based on binned
application entry times and a heuristic ``survival surface'' was generated in
order to get a single overall picture of the survivability of a given disorder. The 
censoring parameters $\tau_z$ varied due to the finite sampling window and 
the fact that more-recent cohorts are not observed for as long a time period
as older cohorts, as depicted in Fig~\ref{fig:fig1}b.
It was also expected that survival by cohort would vary due to differences
in health care administration and treatment between entry cohorts. The
use of the empirical prevalences ($\hat{q}_z$) allowed the accounting for 
variability in disability application volume by sufferers of given disorders, conditional on entry date.

We note that a strong limitation of the presented method lies in its framing in terms of null hypothesis
statistical testing. The $\hat\Theta$ statistic only provides a $P$-value, as opposed to other tests such
as the log-rank test which provide hazard ratios as well. As a trade-off for statistical power, one
is sacrificing interpretability in the form of effect sizes.

 Although the most
direct and natural applications of the method that we have presented here involve
discretely-indexed covariates, it is possible to use this method for continuously-indexed covariates
such as time by employing the binning strategy used in \citet{rasch2014first}. This approach is
particularly fruitful if the sampling windows are coarse and there is clear separation between
cohorts to maintain statistical independence. In this situation, it may be unreasonable to expect to
construct a full continuous surface for survival. Nonetheless, a possible future extension of this
method might involve replacing the sum of Eq.~\ref{eq:conditioningsum} with an integral and 
using statistical regularization tools~\cite{chang2014path} in order to infer true continuously-indexed
surfaces.

\ethics{This section does not apply.}

\dataccess{All data in this manuscript is simulated, with R source code provided in Appendix~\ref{s:datasim}.}

\competing{The authors declare no competing interests.}

\disclaimer{This section does not apply.}

\aucontribute{A.H. and M.H. developed the statistical method. A.H., M.H, and J.C.C. wrote the proof. A.H., M.H, and J.C.C. performed the simulations. J.C.C generated the figures. A.H., M.H, and J.C.C. wrote the manuscript. All authors gave final approval for publication.}

\funding{This work is supported by the Intramural Research Program of the National Institutes of Health Clinical Center and the US Social Security Administration. }
\ack{
The authors would like to thank Dr. Leighton Chan and Dr. Elizabeth Rasch for insightful discussions, guidance and support, Dr. Pei-Shu Ho for help obtaining data.}

  \bibliography{master}
    
    \clearpage

    \newpage
    
  \appendix

\setcounter{secnumdepth}{4}

\section{Proof of the main theorem}\label{s:mainproof}
\setcounter{thm}{0}

To prove the main theorem, we use an empirical process modeling framework to develop the asymptotic
properties of first deterministically proportionally-weighted Kaplan-Meier estimators. We then replace the deterministic proportions with estimates given by the sample prevalences of the cohorts. Here, we 
restate the main theorem and prove it through a series of lemmata.

   \begin{thm}
      \label{thm:main}
      Let $C^{(i)}_{z, t}$ denote the probability that 
      a $z$-type individual has not yet been censored at time $t \geq 0$ (the survival probability
      relative to the occurrence of censoring), and 
      $q^{(i)}_z$ denote the probability that an individual in population $i$ is of cohort 
      $z$, and let $p^{(i)} = n^{(i)}/(n^{(1)}+n^{(2)}).$ Suppose that $\theta^{(1)}_t = \theta^{(2)}_t.$ Then $\hat{\Theta}\xrightarrow{d} N(0,\sigma^{2})$, as $n^{(i)} \to \infty$, with
      \begin{align*}
        \sigma^{2} &= \sum_{i = 1}^{2} (1-p^{(i)})\innp{\sum_{z = 1}^d q^{(i)}_z \phi^{2}_z 
        - \innp{\sum_{z = 1}^d q^{(i)}_z \phi_z}^{2}} \nn
          &\quad- \sum_{z = 1}^d \int_0^{\tau_z}\!\d{S}_{z, t}\,
          W_{z, t}\times\innp{\frac{\phi_{z, t}}{S_{z, t}}}^{2},
      \end{align*}
      where for $0 \leq t \wedge \tau_z$, where $\tau_z$ is the time at which samples of cohort $z$ are censored,
       $\phi_{z, 
      t} = \int_t^{\tau_z}\!\d{s}\, S_{z, s}$, $\phi_z \equiv \phi_{z, 
      0}$, $S_{z,t}$ is the survival function for the pooled data of cohort $z$, and
      \begin{align*}
        W_{z, t} =  \innp{\frac{p^{(1)} C^{(1)}_{z, t-} q^{(2)}_z + p^{(2)} C^{(2)}_{z, t-} 
        q^{(1)}_z}{C^{(1)}_{z, t-} C^{(2)}_{z, t-}}}.
      \end{align*}
      The variance $\sigma^{2}$ may be consistently estimated by
      \begin{align}
        \hat{\sigma}^{2} &= \sum_{i = 1}^{2} (1-p^{(i)})\innp{\sum_{z = 1}^d 
        \hat{q}^{(i)}_z \hat{\phi}_z^{2} - \innp{\sum_{z = 1}^d \hat{q}^{(i)}_z 
      \hat{\phi}_z}^{2}} \nn
        &\quad- \sum_{z = 1}^d \int_0^{\tau_z}\!\d{\hat{S}}_{z, t}\,
        \hat{W}_{z, t}\times\innp{\frac{\hat{\phi}_{z, t}}{\hat{S}_{z, t}}}^{2},
                \label{eqn:sigma2}
      \end{align}
      where for  $0 \leq t \wedge \tau_z$, 
      $\hat{S}_{z, t}$ is the product-limit estimate of the pooled data for cohort $z$, 
            \begin{equation}
      \hat{\phi}_{z, t} = \int_t^{\tau_z}\d{s}\, \hat{S}_{z, s},
      \end{equation}
      $\hat{C}^{(i)}_{z, t}$ is the product-limit estimate associated to the event 
      of censoring for cohort $z$ within population $i$, 
       $ \hat{\phi}_z \equiv \hat{\phi}_{z, 
      0},$ and
      \begin{align}
        \hat{W}_{z, t} =  \innp{\frac{p^{(1)} \hat{C}^{(1)}_{z, t-} \hat{q}^{(2)}_z 
          + p^{(2)} \hat{C}^{(2)}_{z, t-} \hat{q}^{(1)}_z}{\hat{C}^{(1)}_{z, t-} 
        \hat{C}^{(2)}_{z, t-}}}. \label{eq:What}
      \end{align}
    \end{thm}
    
    \begin{proof}[Overview of Proof of Theorem~\ref{thm:main}]
    To prove the main theorem, we turn to the modeling framework that we present in ~\ref{s:model}.
    In general, we proceed by first assuming fixed sample proportions and then extending
    results to random proportions as given by empirical prevalence (Eq~\ref{eq:q_z}).
    The convergence of $\hat{\Theta}$ follows directly from corollary~\ref{cor:conv_sum} and Eq.~\ref{eq:Thetadist}.
     The consistency 
      of $\hat\sigma^2$ follows from theorem 4.2.2 of \cite{gill1980censoring}, which provides for
      weak convergence of the product limit estimator to a Gaussian process, and the 
      Glivenko-Cantelli theorem.

    \end{proof}

  \subsection{Preliminaries and Notation}
    \label{s:prelim}
    
    Given any pair of random elements $X, Y$, we denote equality in 
    a distributional sense by $X \approx Y$. Let $\b{P}$ be a probability 
    measure on the measurable space $\innp{X,\s{A}}$. The empirical measure 
    generated by the sample of random elements $x_1,\dots,x_n$, $n\in\b{N}$ is 
    given by
    \begin{align}
      \label{eqn:empirical_zeasure}
      \b{P}_n
      & = n^{-1}\sum_{i=1}^n\delta_{x_i},
    \end{align}
    where for any $x\in X$, and any $A\in\s{A}$,
    \begin{align}
      \label{eqn:dirac_zeasure}
      \delta_x(A)
      & = \left\{
      \begin{array}{ll}
        1, & x\in A,\\
        0, & x\notin A.
      \end{array}
      \right.
    \end{align}
    Note that alternatively, when needed, one may write 
    $\delta_x(A)$ as the indicator function $1_A(x)$ on the set $A$. 
    Furthermore, in the case that $A=\set{k}$, $k\in\b{Z}$, and 
    $x\in\b{Z}$, we write $\delta_x(A)\equiv\delta_{x,k}$.

    Given $\s{H}$, a class of measurable functions $h:X\to\b{R}$, the empirical 
    measure generates the map $\s{H}\to\b{R}$ given by $h\mapsto \b{P}_n h$,
    where for any signed measure $Q$ and measurable function $h$, we use the 
    notation $Qh=\int\!\d{Q}\,h$. Furthermore, define the $\s{H}$-indexed 
    empirical process $\s{G}_n$ by
    \begin{align}
      \label{eqn:empirical_process}
      \s{G}_n h
      & = \sqrt{n}\innp{\b{P}_n-\b{P}}{h} 
      = \frac{1}{\sqrt{n}}\sum_{i=1}^n\innp{h(x_i)-\b{P}h},
    \end{align}
    and with the empirical process, identify the signed measure 
    $\s{G}_n=n^{-1/2}\sum_{i=1}^n\innp{\delta_{x_i}-\b{P}}$. 

    Note that for a measurable function $h$, from the law of large numbers and 
    the central limit theorem, it follows that $\b{P}_n h\displaystyle\toas\b{P}h$, and 
    $\s{G}_n h\todistribution N{\innp{0,\b{P}\innp{h-\b{P}h}^{2}}}$,
    provided $\b{P}h$ exists and $\b{P}h^{2}<\infty$, and where 
    ``$\todistribution$'' denotes convergence in distribution. In addition to 
    the preceding notation, given the elements $f$, and $f_n$, 
    $n\in\b{N}$, we also denote respectively, convergence in probability and in 
    distribution, of $f_n$ to $f$, by $f_n\toprobability f$.

    For any map $x:\s{H}\to\b{R}^k$, $k\in\b{N}$, define the uniform norm 
    $\norm{x}_{\s{H}}$ by
    \begin{align}
      \label{eqn:uniform_norm}
      \norm{x}_{\s{H}}
      & = \sup{\set{\abs{x(h)}:h\in\s{H}}},
    \end{align}
    and in the case that $\s{H}\subset\b{R}$, write 
    $\norm{\cdot}_\s{H}\equiv\snorm{\cdot}$. 
    A class $\s{H}$ for which $\norm{\b{P}_n-\b{P}}_{\s{H}}\to0$ is called 
    a $\b{P}$-Glivenko-Cantelli class. Denote by $\l^{\infty}(\s{H})$ the class 
    of uniformly bounded functions on $\s{H}$. That is, for a general 
    $k\in\b{N}$, \begin{align*}
      \l^\infty(\s{H}) 
      & = \set{x:\s{H}\to\b{R}^k : \norm{x}_{\s{H}} < \infty}.
    \end{align*}
    If for some tight Borel measurable element $\s{G}\in\l^\infty(\s{H})$, 
    $\s{G}_n\todistribution\s{G}$, in $\l^\infty(\s{H})$, we say that $\s{H}$ is 
    a $\b{P}$-Donsker class.

  \subsection{Empirical process framework}
    \label{s:model}
 
    To prove Theorem~\ref{thm:main}, we turn to an empirical modeling framework that
    will provide us the asymptotic statistics of the weighted product limit estimator.    
    Consider a closed particle-system, such that according to a predefined set 
    of characteristics, the system can be subdivided into mutually exclusive 
    subsystems. 
    
    Each particle corresponds to the observed state of a particular individual in a fixed population cohort.
     Note that we will restrict this discussion to only a single population
    of particles. These arguments will extend to multiple populations as
    mentioned in this manuscript by treating separate populations as 
    independent.
    
     At any given time $t\geq 0$, each particle will have exactly one 
    associated state $x$ in the set $\b{Z}_4$, referring 
    respectively to states of 
\begin{equation}
\begin{matrix} 
    0& \textrm{dormancy} \\
    1 & \textrm{activity}\\
    2 & \textrm{inactivity}\\
    3 & \textrm{censored}.
\end{matrix}
\end{equation}
    Assume that the path of any particle is statistically dependent upon its
    particular subsystem, and that given the respective subsystems of any two 
    particles, their resulting paths are statistically independent. Assume further that at
    a reference time $t=0$, all particles enter into the active state $(x=1),$
    and that particles are considered dormant for all $t<0$.

    Let $d \in \b{N}$ and $\tau \in (0, \infty)$ be fixed. We will assume the 
    existence of a collection of individuals $\Gamma$, assumed to be infinite in 
    size, where each individual $\gamma \in \Gamma$ exhibits a c\`adl\`ag 
    path-valued state $x_t^\gamma$, for $t \geq 0$. For 
    each $\gamma \in \Gamma$, $x^\gamma_t$ is determined by the 
    individuals particle type $z^\gamma$ and a random jump time $\xi^\gamma$. 
    The particle type $z^\gamma$ is distributed in the population through
    the probability mass $\b{P}(z^\gamma= z) = q_z$, where $\mathbf{q} = (q_1, \dots, q_d) \in (0, 1)^d$ satisfies 
    $\sum_{z=1}^d q_z = 1$.
    %
%
    Let $S_t = (S_{1,t}, \dots, S_{d,t})$ be the survival vector $S_{z, t} = \b{P}\set{T_z 
    > t}$, which is assumed continuous for $t \geq 0$. Suppose that it is 
    desired to understand the event probabilities for randomly selected $\gamma 
    \in \Gamma$, unconditional on subgroup membership. 
    We assume that members of each cohort are in the inactive (0) state
    at times $t<0$.

    Given a random sample $\gamma_1, \dots, \gamma_n$, $n \in \b{N}$ of 
    individuals, let $\mathbf{n} = (n_1, \dots, n_d)$ and
    \begin{equation}
    n = \sum_{z=1}^d n_z
    \end{equation}
     where $n_z$ is the random number of drawn 
    individuals of cohort $z$.     
     In considering the event time probabilities of each
    subgroup, the random number of particles excludes the use of many 
    well established results in survival analysis. Therefore, we begin with 
    a somewhat restricted framework, and assume a known number of initial 
    individuals of each type. 
    
    Assume the sample 
    contains  a \textbf{known} number $n_z = a_z n$, $a_z\in(0,1)$,
     of individuals of cohort $z$, and let $\mu_{j, z, t}^{n_z}\geq 0$ be 
    the number of the   cohort $z$ individuals who are in state $j \in 
    \b{Z}_4$ at time $t$, so that
    
    \begin{equation}
    \sum_{j=0}^3 \mu_{j,z,t}^{n_z} = n_z
    \end{equation}
    is conserved. Also, we assume that there exists $\tau_z < \infty$ when all
    particles either become inactive or censored so that $\tau_z$ is
    the infimum time where the condition
    \begin{equation}
    n_z = \mu_{2,z,t}^{n_z} + \mu_{3,z,t}^{n_z} \qquad \forall t>\tau_z
    \end{equation}
    holds.
    
    For the sample of size $n_z$, we
    denote the $z$-type cumulative hazard by $\Lambda_{z, t}$ and respectively 
    define the $z$-type cumulative hazard and survival estimates by
\begin{align}
\hat{\Lambda}_{z, t}^{n_z} &= \int_0^t \frac{\dd \mu^{n_z}_{ 2,z, s}}{\mu^{n_z}_{ 1,z, s-}} \label{eq:Lambdahat}\\
 \hat{S}^{n_z}_{z, t} &= \prod_{s \leq t} \innp{1 - \dd \hat{\Lambda}^{n_z}_{z, s}}. \label{eq:Shat}
 \end{align}
    Define further \[ B^{n_z}_{z, t} = \sqrt{n_z}\, \frac{\hat{S}^{n_z}_{z, t} - S_{z, t}}{S_{z, 
    t}}\] and note that $\hat{S}^{n_z}_{z, t} = \hat{S}^{n_z}_{z, \tau_z}$ and 
    ${B}^{n_z}_{z, t} = {B}^{n_z}_{z, \tau_z}$ for all $t \geq \tau_z$.

    From \cite{gill1980censoring}, it follows that $\set{B^{n_z}_{z, t} : t \geq 0}$ is 
    a mean-zero square-integrable martingale with Meyer bracket process 
\begin{equation}\label{eq:meyer}
      \inn{B^{n_z}_{z,t}, B^{n_z}_{w,t}}_t
      = \delta_{zw} n_z \int_0^{t \wedge \tau_z} 
      \dd  \Lambda_{z, s} \innp{\frac{\hat{S}^{n_z}_{z, s-}}{S_{z, s}}}^{2}
      \frac{1_{\set{\mu^{n_z}_{ 1,z, s-}>0}}}{\mu^{n_z}_{ 1,z, s-}},
\end{equation}
where $t \wedge \tau_z = \min\{ t, \tau_z \},$ and $\delta_{(\cdot,\cdot)}$ is the
Kroenicker delta function.

\subsection{Convergence Theorems}
    \label{s:conv}

    In order to guarantee convergence of the estimator, we make the following 
    assumptions (based upon an initially known sample size distribution $\mathbf{n}$).
    \begin{asm}
      \label{asm:sample_size}
      We assume that the initial sample is chosen large enough to ensure 
      that individuals of cohort $z$, at state ${1}$ (active), exist at all points $t \in 
      [0, \tau_z]$, $z \in \set{1, \dots, d}$. That is,
      \begin{align*}
        \inf_{z \in \b{N}_d} \mu^{n_z}_{1, z, \tau_z-}
        & > 0, \quad \mathrm{a.s.}
      \end{align*}
    \end{asm}
    Since any survival function is monotone, an immediate result that 
    follows from the above is assumption is 
    \begin{align}
      \label{bound:S}
      c < S_{z, \tau_z} \leq S_{z, t} \leq 1, \quad t \geq 0,
    \end{align}
    for some constant $c>0$.
    \begin{asm}
      \label{asm:sample_size_2}
      It is assumed that as $n$ becomes large, the sample size for each 
      individual type will grow to infinity. That is, 
      \begin{align*}
        \lim_{n \to \infty} \inf_{z \in \b{N}_d, a \in V}
        \mu^{n a_z}_{1, z, \tau_z-} = \infty, \quad \mathrm{a.s.}
      \end{align*}
    \end{asm}
    \begin{asm}
      \label{asm:sample_size_3}
      For each $z \in \set{1, \dots, d}$ there exists a non-increasing 
      continuous function $m_z:[0, \infty) \to (0, 1]$ such that
      \begin{align*}
        \lim_{n \to \infty} \sup_{t \geq 0}
        \abs{\frac{\mu^{n a_z}_{1,  z, t}}{n a_z} - m_{z,t}}=0 \quad \mathrm{a.s.}
      \end{align*}
    \end{asm}

    Note that in the case of fixed censoring, that is, in the case that 
    censoring exists only at time $\tau$, the above is satisfied by $m_{z, t} 
    = S_{z, t}$. In the general case, $m_{z, t}$ can be seen as the probability 
    that an individual of cohort $z$ has not yet left state $1$. That is, $m_{z, 
    t}$ is the probability that an individual has not left due to censoring or 
    death by time $t$, and so $m_{z, t} = S_{z, t} C_{z, t-}$,
    where $C_{z, t}$ is the probability that censoring has not occurred by time 
    $t$.  

    To prove the main theorem, we now present a series of lemmata.
    \begin{lem}
      \label{lem:conv_Z_first_term}
      If $\hat{q}$ is defined is in Eq.~\ref{eq:q_z} and $\hat{S}_{z, s-}^{n_z}$ is defined as in Eq.~\ref{eq:Shat}, then 
      \begin{align*}
        \sqrt{n} \sum_{z = 1}^d (\hat{q}_z - q_z) 
        \int_0^{t \wedge \tau_z} \!\d{s}\, 
        \innp{\hat{S}_{z, s-}^{n\hat{q}_z} - S_{z, s}} \toprobability 0, 
      \end{align*}
      as $n\to\infty$, uniformly in $t \geq 0$. 
    \end{lem}
    \begin{proof}
      It is claimed that to prove the statement of the lemma, it 
      suffices to show that 
      \begin{align}
        \label{conv:B_N}
        \sup_{t \geq 0} 
        \innp{\frac{\hat{S}_{z, t-}^{n\hat{q}_z}  - S_{z, t}} 
        {S_{z, t}}}^{2} \toprobability 0, 
      \end{align}
      uniformly in $t \geq 0$, for each $z = 1, \dots, d$. 

      Indeed, for if the above holds, then
      \begin{align*}
        \int_0^{t \wedge \tau_z} \!\d{s}\, 
        \innp{\hat{S}_{z, s-}^{n\hat{q}_z} - S_{z, s}} \toprobability 0, 
      \end{align*}
      uniformly in $t \geq 0$. Since the central limit theorem implies 
      that $\sqrt{n} (\hat{q}_z - q_z) \todistribution N(0, q_z(1 - q_z))$, 
      each term in the sum would converge in probability to $0$, 
      uniformly in $t \geq 0$.

      And so, if $\b{E}_{N}$ denotes the expectation given $N$, 
      we have that 
      \begin{align*}
       \lefteqn{ \b{E} \innp{\frac{\hat{S}_{z, t-}^{n\hat{q}_z}  - S_{z, t}} 
        {S_{z, t}}}^{2}
         = \b{E} \frac{1}{n\hat{q}_z } \b{E}_{n\hat{q}_z } \innp{B_{z, t}^{n\hat{q}_z }}^{2}  }\\
        & \qquad\qquad= \b{E} \frac{1}{n\hat{q}_z} \b{E}_{n\hat{q}_z} n\hat{q}_z 
        \int_0^{t \wedge \tau_z}\!
        \frac{\d{\Lambda_{z, s}}}{\mu^{n\hat{q}_z}_{ 1,z, s-}}
        \innp{\frac{\hat{S}^{n\hat{q}_z} _{z, s-}}{S_{z, s}}}\\
        & \qquad\qquad= \b{E}\int_0^{t \wedge \tau_z}\!
        \frac{\d{\Lambda_{z, s}}}{\mu^{n\hat{q}_z} _{ 1,z, s-}}
        \innp{\frac{\hat{S}^{n\hat{q}_z} _{z, s-}}{S_{z, s}}}\\
        & \qquad\qquad\leq C \b{E} 
        \innp{\mu^{n\hat{q}_z}_{1, z, \tau_z}}^{-1},
      \end{align*}
      for some arbitrary constant $C$. From Lenglart's inequality (cf. 
      \cite{lenglart1977relation}),
      \begin{align*}
        \b{P}\! \set{\sup_t \innp{\frac{\hat{S}_{z, t-}^{n\hat{q}_z} 
            - S_{z, t}}
        {S_{z, t}}}^{2} > \epsilon}
        & \leq  \frac{\eta}{\epsilon} + \b{P}\!
        \set{\mu^{n\hat{q}_z} _{1, z, \tau_z-}
        < \frac{C}{\eta}},
      \end{align*}
      for any arbitrary $\eta, \epsilon > 0$. Therefore, from 
      Assumption~\ref{asm:sample_size_2}, since $n_z \to \infty$ a.s., 
      the desired result follows.
    \end{proof}

Turning momentarily to the situation where there are two populations denoted by superscripts $(1)$, and $(2)$,
    for any $t \geq 0$, define
    \begin{align*}
      \hat{\Theta}_t^\delta &= \sqrt{\frac{n^{(2)}}{n^{(1)}+n^{(2)}}}
      \int_0^{t \wedge \tau} \!\d{s}\, \sqrt{n^{(1)}} (\hat{\theta}^{(1)}_{s-} 
      - \theta^{(1)}_s)  \nn
      &\quad- \sqrt{\frac{n^{(1)}}{n^{(1)}+n^{(2)}}}
      \int_0^{t \wedge \tau} \!\d{s}\, \sqrt{n^{(2)}} (\hat{\theta}^{(2)}_{s-} 
      - \theta^{(2)}_s),
    \end{align*}
    noting that setting $\theta_s^{(1)} = \theta_s^{(2)}$ recovers our test statistic 
    of Eq.~\ref{eq:Theta}.
    For a general survival function $\theta$, with respective estimate 
    $\hat{\theta}$, define $\hat{Y}_t$ by 
    \begin{align}
      \label{eqn:estimate_integral}
      \hat{Y}_t = \int_0^{t \wedge \tau} \!\d{s}\, \sqrt{n} 
      \innp{\hat{\theta}_{s-} - \theta_s}, \quad t \geq 0.
    \end{align}
    If the process $\hat{Y}$ converges in distribution to some $Y \sim N(0, 
    \sigma^2)$, since $n^{(i)} / (n^{(1)}+n^{(2)})$ converges to $p^{(i)}$, $i = 1, 2$, it 
    follows that 
    \begin{align}
      \hat{\Theta}^\delta_t \todistribution \sqrt{p^{(2)}} Y^{(1)}_t - \sqrt{p^{(1)}} Y^{(2)}_t 
      \approx N(0, p^{(2)} \sigma^2_{1} + p^{(1)} \sigma^2_{2}).\label{eq:Thetadist}
    \end{align}

Now we turn to analysis under a single population, dropping the superscripts.
    Note that $\hat{Y}_t = \sum_{z = 1}^d \hat{Z}_{z, t}$, where 
    \begin{align}
      \label{eqn:Z_n_z}
      \hat{Z}_{z, t} 
      & = \sqrt{n} \int_0^{t \wedge \tau_z} 
      \!\d{s}\, \innp{\hat{q}_z \hat{S}_{z, s-}^{n\hat{q}_z} - q_z 
      S_{z,s}}\\
      & = \sqrt{n} (\hat{q}_z - q_z) \int_0^{t \wedge \tau_z}\!\d{s} 
      \, \innp{\hat{S}_{z, s-}^{n\hat{q}_z} - S_{z, s}} \nonumber \\ 
      & \qquad + \sqrt{n} (\hat{q}_z - q_z) 
      \int_0^{t \wedge \tau_z}\!\d{s}\, 
      S_{z, s} \nonumber \\ 
      & \qquad + \sqrt{n} q_z \int_0^{t \wedge \tau_z}\!\d{s}\, 
      (\hat{S}_{z, s-}^{n\hat{q}_z}  - S_{z, s}) \nonumber 
    \end{align}
    Therefore, if it can be shown that 
    \begin{align*}
      \sqrt{n} \sum_{z = 1}^d (\hat{q}_z - q_z) 
      \int_0^{t \wedge \tau_z} \!\d{s}\, 
      \innp{\hat{S}_{z, s-}^{n\hat{q}_z}  - S_{z, s}} \toprobability 0, 
    \end{align*}
    uniformly in $t$, then convergence of $(\hat{Y}_t : t \geq 0)$ is 
    dependent only upon the convergence of the $d$-dimensional 
    vector-valued process $\hat{\zeta}(\hat{q})$ given by 
    \begin{align}
      \label{eqn:zeta_n}
      \hat{\zeta}_{z, t}(a) 
      & = \sqrt{n} (\hat{q}_z - q_z) \int_0^{t \wedge \tau_z}\!d{s}\,
      S_{z, s} \nn
       &\quad+ \sqrt{n} q_z \int_0^{t \wedge \tau_z}\!\d{s}\, 
      (\hat{S}_{z, s-}^{na_z} - S_{z, s}),
    \end{align}
    with $a = (a_1, \dots, a_d) \in (0, 1)^d$ chosen in a sufficiently small 
    neighborhood $V$ of $q$. This decomposition will thus lead to the main 
    theorem. To show the desired convergence of $\hat{\zeta}_t(\hat{q})$, we 
    first focus on convergence of $\hat{\zeta}_t(a)$. 

    Let $\phi_{z, t} = \int_t^{\tau_z}\!\d{s}\, S_{z, s}$ and 
    write $\hat{\zeta}_t(a) = \hat{\zeta}^1_t + \hat{\zeta}^2_t(a)$, where 
    
    \begin{align}
      \label{eqn:zeta_1}
      \hat{\zeta}^1_{z, t} = \sqrt{n} \innp{\hat{q}_z - q_z}
      \int_0^{t \wedge \tau_z}\!(-\d{\phi_{z, s}}),
    \end{align}
    and
    \begin{align}
      \label{eqn:zeta_2}
      \hat{\zeta}^2_{z, t}(a) = \frac{q_z}{\sqrt{a_z}} \int_0^{t \wedge \tau_z}\!(-\d{\phi_{z, s}}) B_{z, s}^{na_z},
    \end{align}

    \begin{lem}
      \label{lem:conv_zeta_parts}
      Suppose that $\set{\hat{\zeta}^1_t(a) : t \geq 0}$ 
      and $\set{\hat{\zeta}^2_t(a) : t \geq 0}$ are the processes respectively 
      defined by equations (\ref{eqn:zeta_1}) and (\ref{eqn:zeta_2}), and that 
      $\tilde{B}$ is the $d$-dimensional mean-zero Gaussian process defined by 
      \begin{align}
        \inn{\tilde{B}_z, \tilde{B}_w}_t = \delta_{z,w} \int_0^{t \wedge \tau_z}\!\frac{\d{\Lambda_{z, s}}}{S_{z, s} 
        C_{z, s-}}.\label{eq:Btilde}
      \end{align}
      Then $\hat{\zeta}^1_t \todistribution \zeta^1_t  
        \mathrm{and} \quad 
        \hat{\zeta}^2_t(a) \todistribution \zeta^2_t(a),$ 
        in the space of compactly supported functions $\s{D}_{\b{R}^d}[0, \infty)\,\,  \mathrm{as}\,\,
        n \to \infty, $
      for each $a \in V$, where 
      $\zeta^1_t = (\zeta^1_{1, t}, \dots, \zeta^1_{d, t})$ 
      is the mean-zero square-integrable Gaussian process defined by
      \begin{align}
        \label{eqn:zeta_1_bracket}
        & \inn{\zeta^1_z, \zeta^1_w}_t \\ & = -q_z q_w \innp{\int_0^{t \wedge \tau_z} \!\d{s}\, S_{z, s}}
        \innp{\int_0^{t \wedge \tau_w}\!\d{s}\, S_{w, s}} \nn
        &\quad + \delta_{z, w} q_z \innp{\int_0^{t \wedge \tau_z}\!\d{s}\, 
        S_{z, s}}^{2},
        \nonumber
      \end{align}
      and $\zeta^2_t(a) = (\zeta^2_{1, t}(a), \dots, \zeta^2_{d, t})$ 
      is given by 
      \begin{align}
        \label{eqn:zeta_2_bracket}
        \zeta^2_{z, t}(a) = \frac{q_z}{\sqrt{a_z}} \innp{\int_0^{t \wedge \tau_z}\!\d{\tilde{B}}_{z, s} \phi_{z, s} - \phi_{z, t \wedge \tau_z} \tilde{B}_{z, t \wedge \tau_z}}
      \end{align}
      The processes $\hat{\zeta}^1$ and $\hat{\zeta}^2(a)$ are independent, and 
      there exist a Skorohod representations such that \begin{align*}
        \sup_{t \geq 0} 
        \abs{\hat{\zeta}^1_{z, t} - \zeta^1_{z, t}} \to 0,
      \end{align*}
      and 
      \begin{align*}
        \sup_{t \geq 0, a \in V} 
        \abs{\hat{\zeta}^2_{z, t}(a) - \zeta^2_{z, t}(a)} \to 0,
      \end{align*}
      almost surely as $n \to \infty$.
    \end{lem}
    \begin{proof}
      To begin note that independence follows immediately from the 
      independence of the respective limiting processes. Since $\mathbf{n}$ is 
      a multinomial random variable, (\ref{eqn:zeta_1_bracket}) follows 
      from the central limit theorem. In the case of $\hat{\zeta}^2_t(a)$, 
      we first consider 
      $B_{z,t}^{na_z}$.
     
      An application of Lenglart's inequality, very similar to that in the 
      proof of Lemma~\ref{lem:conv_Z_first_term}, along with 
      Assumption~\ref{asm:sample_size_2}, shows that 
      \begin{align*}
        \sup_{a \in V, t \geq 0} 
        \abs{\hat{S}_{z, t-}^{na_z} - S_{z, t}} \toprobability 0, 
        \quad \mathrm{as}\,\, n \to \infty.
      \end{align*}
      Moreover, from Assumption~\ref{asm:sample_size_3}, 
      \begin{align*}
        \sup_{a \in V, t \geq 0} 
        \abs{\frac{n a_z}
        {\mu^{na_z}_{1, z, t-}} - \frac{1}{m_{z, t}}} \toprobability 0, \quad \mathrm{as}\,\, 
        n \to \infty.
      \end{align*}
      It follows that 
      \begin{align*}
        \frac{n a_z}{\mu^{na_z}_{1, z, t-}}
        \innp{\frac{\hat{S}_{z, t-}^{na_z}}{S_{z, t}}}^{2}
        \toprobability \frac{1}{m_{z, t}},
      \end{align*}
      uniformly in $t \geq 0$, and since  $m_{z, t} = S_{z, t} C_{z, t-}$,
      \begin{align*}
        \inn{B_{z,t}^{na_z}, B_{w,t}^{na_z}}_t \toprobability \delta_{z,w} \int_0^{t \wedge \tau_z}\! \frac{\d{\Lambda_{z, s}}} 
        {S_{z, s}C_{z, s-}}.
      \end{align*}  

      Therefore, from theorem 4.2.1 of \cite{gill1980censoring}, $B^{na_z}_{z,t}
      \todistribution \tilde{B}_{z,t}$, and there exists a Skorohod representation of 
      $B^{na_z}_{z,t}$ such that
      \begin{align*}
        \sup_{t \geq 0, a \in V} \abs{B_{z, t}^{na_z} - \tilde{B}_{z, t}} 
        \to 0,
      \end{align*}
      almost surely as $n \to \infty$. Since almost sure convergence of 
      $B_{z, t}^{na_z} $ implies almost sure convergence of bounded functionals of 
      $B_{z, t}^{na_z} $, the desired convergence of $\hat{\zeta}^2(a)$ follows 
      from Theorem 2.1 of \cite{gill1981large}.
    \end{proof}

    \begin{coro}
      \label{cor:conv_sum_1}
      If the process $\hat{\zeta}(a) = \set{\hat{\zeta}_t(a)}$ is defined by 
      equation (\ref{eqn:zeta_n}), then 
      \begin{align}
        \sum_{z = 1}^d \hat{\zeta}_z(a) 
        & \todistribution \sum_{z = 1}^d \zeta_{z, t}(a) \\
        & = \sum_{z = 1}^d \zeta^1_{z, t} + \zeta^2_{z, t}(a)\nonumber.
      \end{align}
    \end{coro}
    \begin{proof}
      From the previous theorem we may assume that 
      $\hat{\zeta}^1_{z, t} \to \zeta^1_{z, t}$ and 
      $\hat{\zeta}^2_{z, t}(a) \to \zeta^2_{z, t}(a)$ almost surely, 
      uniformly for $a \in V$ and $t \geq 0$. Therefore 
      \begin{align*}
        \hat{\zeta}_t(a) \to \zeta_t(a)
      \end{align*}
      almost surely, uniformly for $a \in V$ and $t \geq 0$. The statement 
      of the theorem then follows from theorem 5.1 of \cite{billingsley2013}. 
    \end{proof}
    
    Since $\mathbf{n}/n \toprobability \mathbf{q}$, from Theorem 4.4 of \cite{billingsley2013}    
    \begin{align*}
      \innp{\frac{\mathbf{n}}{n}, 
      \set{\sum_{z = 1}^d \hat{\zeta}_{z, t}(a) : t \geq 0}} 
      \todistribution \innp{\mathbf{q}, 
      \set{\sum_{z = 1}^d \zeta_{z, t}(a) : t \geq 0}}. 
    \end{align*}
    
    Define the map 
    $g : V \times \l^\infty(V \times [0, \infty)) \to 
      \l^\infty([0, \infty))$ by $g(a, f) = f(a, \cdot)$, then 
    \begin{align*}
      \sum_{z = 1}^d \zeta_{z, t}\innp{\frac{\mathbf{n}}{n}} 
      & = g\innp{\frac{\mathbf{n}}{n}, \sum_{z = 1}^d \zeta_z}.
    \end{align*}
    Furthermore, if for any $(a_1, f_1), (a_2, f_2) \in V \times \l^\infty(V 
    \times [0, \infty))$ we have that
    \begin{align*}
      \abs{a_1 - a_2} + \sup_{a \in V, t \geq 0} 
      \abs{f_1(a, t) - f_2(a, t)}
      & < \delta
    \end{align*}
    for some $\delta > 0$, then
    \begin{align*}
      \sup_{t \geq 0} & \abs{g(a_1, f_1)(t) - g(a_2, f_2)(t)} \\ 
                      & = \sup_{t \geq 0} \abs{f_1(a_1, t) - f_2(a_2, t)} \\
                      & \leq \sup_{t \geq 0} \abs{f_1(a_1, t) - f_1(a_2, t)}
                      + \sup_{t \geq 0} \abs{f_1(a_2, t) - f_2(a_2, t)}.
    \end{align*}
    Therefore, $g$ is a continuous at any $(a, f)$ such that 
    $f$ is continuous at $a$, uniformly in $t$. It thus follows from the 
    continuous mapping theorem (cf. \cite{van1996new}) that if $a \mapsto 
    \sum_{z = 1}^d \zeta_{z, t}(a)$ is continuous, uniformly in $t$, then 
    \begin{align}
      \label{conv:g}
      g\innp{\frac{\mathbf{n}}{n},\sum_{z = 1}^d \hat{\zeta}_z} 
      \todistribution g\innp{\mathbf{q},\sum_{z = 1}^d \zeta_z}.
    \end{align}
    
    \begin{lem}
      \label{lem:cont_a}
      If $\set{\zeta_t(a) : t \geq 0}$ is defined as in 
      Corollary~\ref{cor:conv_sum_1}, then the map 
      \begin{align*}
        a \mapsto \sum_{z = 1}^d \zeta_{z, t}(a) 
      \end{align*}
      is continuous for $a \in V$, uniformly in $t \geq 0$.
    \end{lem}
    \begin{proof}
      For any $a, b \in V$, it follows that
      \begin{align*}
        & \sum_{z = 1}^d \zeta_{z, t}(a) - \sum_{z = 1}^d \zeta_{z, t}(b)\\
        & = \sum_{z = 1}^d 
        q_z \innp{\frac{1}{\sqrt{a_z}} - \frac{1}{\sqrt{b_z}}} \nn
        &\quad\times \innp{\phi_{z, 
          t \wedge \tau_z} \tilde{B}_{z, t \wedge \tau_z}
          - \int_0^{t \wedge \tau_z}\!\d{\tilde{B}_{z, s}} \phi_{z, 
        s}}.
      \end{align*}
      Since $S_{\tau_z} > 0$ for all $z$, from Doob's martingale 
      inequality (cf.~\cite{karatzas2012brownian}),
      \begin{align*}
        \b{E} \sup_{t \geq 0} \innp{\sum_{z = 1}^d \zeta_{z, t}(a) 
          - \sum_{z = 1}^d \zeta_{z, t}(b)}^{2} \leq  C \sum_{z = 1}^d 
        \innp{\frac{1}{\sqrt{a_z}} - \frac{1}{\sqrt{b_z}}}^{2},
      \end{align*}
      for some arbitrary constant $C$. For each $z \in \b{N}_d$, since 
      $a_z$ and $b_z$ are sufficiently close to $q_z \in (0, 1)$, it 
      follows that there exists some $\delta > 0$ such that 
      $a_z \wedge b_z > \delta$. Therefore, 
      \begin{align*}
        \innp{\frac{1}{\sqrt{a_z}} - \frac{1}{\sqrt{b_z}}}^{2} 
        & = \frac{1}{a_z b_z} (\sqrt{a_z} - \sqrt{b_z})^{2} \\ 
        & \leq \delta^{-2}(\sqrt{a_z} - \sqrt{b_z})^{2}
        \innp{\frac{\sqrt{a_z} + \sqrt{b_z}}{\sqrt{a_z} + \sqrt{b_z}}}^{2} \\
        & \leq \frac{1}{4 \delta^3} (a_z - b_z)^{2}.
      \end{align*}
      Combining the above two results gives 
      \begin{align*}
        \b{E} \sup_{t \geq 0} \innp{\sum_{z = 1}^d \zeta_{z, t}(a) - 
        \sum_{z = 1}^d \zeta_{z, t}(b)}^{2}
        & \leq  C \abs{a - b}^{2}, 
      \end{align*}
      and so, by Kolmogorov's continuity criterion (cf. \cite{karatzas2012brownian}), the 
      desired result follows.
    \end{proof}

    The above lemma, along with the argument immediately preceding, gives 
    the following.

    \begin{thm}
      \label{thm:conv_total}
      Let $\sum_{z = 1}^d \zeta_{z, t}^n(\cdot)$ and 
      $\sum_{z = 1}^d \zeta_{z, t}(\cdot)$ be defined as in 
      Corollary~\ref{cor:conv_sum_1}, then 
      \begin{align}
        \sum_{z = 1}^d \hat{\zeta}_{z, t}\innp{\frac{\mathbf{n}}{n}} 
        \todistribution \sum_{z = 1}^d \zeta_{z, t}(\mathbf{q}), 
        \quad \mathrm{in}\,\, \s{D}_{\b{R}}[0, \infty),\,\, \mathrm{as}\,\,
          n \to \infty.
      \end{align}
    \end{thm}

    \begin{coro}
      \label{cor:conv_sum}
      If $\hat{\zeta} = \sum_{z = 1}^d \zeta_{z, \tau_z}(q)$, then 
      \begin{align*}
        \hat{\zeta} \sim N(0, \sigma^2),
      \end{align*}
      where 
      \begin{align*}
        \sigma^2&= \sum_{z = 1}^d q_z \phi_{z, 0}^2 - \innp{\sum_{z = 1}^d q_z 
        \phi_{z, 0}}^{2} \nn
       &\qquad\qquad - \sum_{z = 1}^d q_z \int_0^{\tau_z}\! \frac{\d{S}_{z, 
        t}}{C_{z, t-}} \innp{\frac{\phi_{z, t}}{S_{z, t}}}^{2}
      \end{align*}
    \end{coro}
    \begin{proof}
      Note that when $t = \tau_z$, we have 
      \begin{align*}
        \zeta_{z, \tau_z}(q) & = \zeta^1_{z, \tau_z} 
        + \sqrt{q_z} \int_0^{\tau_z}\!\d{\tilde{B}}_{z, t}\,\phi_{z, 
        t},
      \end{align*}
      which are independent and normally distributed, implying that 
      $\hat{\zeta}$ is also normally distributed. Furthermore
      \begin{align*}
        \b{E}\hat{\zeta}^2 & = \sum_{z = 1}^d \innp{\zeta^1_{z, \tau_z}
        + \sqrt{q_z} \int_0^{\tau_z}\!\d{\tilde{B}}_{z, t} \phi_{z, t}}^{2}\\
        & \quad + \sum_{\substack{z, w = 1 \\ z \neq w}}^d 
        \innp{\zeta^1_{z, \tau_z} + \sqrt{q_z} \int_0^{\tau_z}\!\d{\tilde{B}}_{z, t} \phi_{z, t}} \nn
        &\qquad\qquad\qquad\times
        \innp{\zeta^1_{w, \tau_w} + \sqrt{q_z} \int_0^{\tau_w}\!\d{\tilde{B}}_{w, t} \phi_{w, t}} \\ 
        & = \sum_{z = 1}^d 
            \innp{\b{E}\innp{\zeta^1_{z, \tau_z}}^{2} + \b{E} q_z 
            \innp{\int_0^{\tau_z}\!\d{\tilde{B}}_{z, t} \phi_{z, 
          t}}^{2}} \nn
        &\qquad\qquad - \sum_{\substack{z, w = 1 \\ z \neq w}}^d 
        \b{E}\zeta^1_{z, \tau_z} \zeta^1_{w, \tau_w} \\ 
        & = \sum_{z = 1}^d \innp{q_z (1 - q_z) \phi_{z, 0}^{2}
        + q_z \int_0^{\tau_z}\!\d{\Lambda_{z, t}}
        \frac{\phi^{2}_{z, t}}{S_{z, t} C_{z, t-}}} \nn
       &\qquad\qquad - \sum_{\substack{z, w = 1 \\ z \neq w}}^d 
        q_z q_w \phi_{z, 0} \phi_{w, 0},
      \end{align*}
      which after recombining the final terms, gives the desired result. 
    \end{proof}
    
\section{Computation}

\subsection{Installation of R package}

The following code installs the R package from github sources
\begin{lstlisting}
# install devtools if not already installed
install.packages("devtools")
library(devtools)
install_github("joshchang/calonesurv")
library(calonesurv)
\end{lstlisting}

\subsection{Simulation of data used in this manuscript}
\label{s:datasim}
We simulated draws from the populations mentioned in the main text using the following R code:

\begin{lstlisting}
library(survival)
library(calonesurv)

p = 0.75
n = 20
samples = 10000
tvals = seq(0,1/explambda*4,by=0.05)

explambda = 1/4
wshape = 1
wscale = 5

mix_results <- matrix(nrow = length(tvals), ncol = samples)
pure_results <- matrix(nrow = length(tvals), ncol = samples)

for(i in 1:samples){
  n1 = sum(rbinom(n,1,p) )
  n2 = n-n1
  #while(n2<2){
  #  n1 = sum(rbinom(n,1,p) )
  #  n2 = n-n1
  #}
  t1 = sort(rexp(n1,explambda))
  t2 = sort(rweibull(n2,wshape,wscale))
  
  km1 = cadlag(survfit(Surv(t1)~1))
  km3 = cadlag(survfit(Surv(c(t1,t2))~1))
  pure <- cadlag(km3$time,c(1,km3$surv))
  
  if(n2>0){
    km2 = cadlag(survfit(Surv(t2)~1))
    mix <- n1/n*km1 + n2/n*km2
    mix_results[,i] <- mix(tvals)
  }
  else{
    mix_results[,i] <- pure(tvals)
  }
  
  pure_results[,i] <- pure(tvals)
}




\end{lstlisting}

\subsection{Real-world example}
\label{sec:realworld}
\begin{lstlisting}
require(coin)
require(survival)
source("Theta.R")
library(data.table)
library(ggplot2)

surv_data = with(subset(survival::lung,ph.ecog %in% 0:2), 
                 data.frame(population = sex, 
                            censor = as.numeric(status==1), 
                            time = time, cohort = ph.ecog ))
                            
out = Theta_hat(surv_data)
print(out)
print(confint(out))
print(pvalue.Theta_hat(out))

\end{lstlisting}

\subsubsection{Simulations for examining the sampling distribution of $\hat\Theta$}

\begin{lstlisting}

J = 10000
N = c(40,80,136)


typeone = data.table("method" = rep(0,J*length(N)*3), "n" = rep(0,J*length(N)*3), "P" = rnorm(J*length(N)*3))
Thetavals = data.table("n" = rep(0,J*length(N)), "Theta" = rep(0,J*length(N)))

# Simulations for evaluating Type-I error

k = 0
j = 1
for(n in N){
  i = 1
  while(i< J*3){
    indices = sample(1:136,n)
    s_data = subset(surv_data,population==1)[indices,]
    s_data$population[1:(n/2)] = 2
    
    tryCatch({
      out = Theta_hat(s_data)
      p_Theta_hat = pvalue.Theta_hat(out)
      p_lr = pvalue(logrank_test(Surv(time,1-censor)~as.factor(population), data = s_data))
      fit = survfit(Surv(time,1-censor)~as.factor(population), data = s_data)
      t1 = ten(fit)
      invisible(comp(t1))
      p_w = attr(t1,'lrt')$pNorm[2]
      Thetavals[j] = list( "n" =n , "Theta" = as.numeric(out))
      typeone[J*3*k + i,] = list("method"=1,"n"=n,"P" = p_Theta_hat)
      typeone[J*3*k + i+1,] = list("method"=2,"n"=n,"P" = p_w)
      # logrank is 3
      typeone[J*3*k + i+2,] = list("method"=3,"n"=n,"P" = p_lr)
      i = i + 3
      j = j + 1
    },error = function(err) print(err))
    
  }
  k = k+1
}
pdf("fig6.pdf",family="CM Roman", width=8, height=4)
typeone$method = as.factor(typeone$method)
levels(typeone$method) = c("Our Method", "Wilcoxon signed rank","Log-rank")
pp = ggplot(typeone, aes(factor(n),P))
print(pp + geom_violin(aes(fill=factor(method)))+geom_hline(aes(yintercept=0.05),linetype='dashed')+ ylim(1e-4,1)+
  scale_y_continuous(trans='log2',breaks = c(0.0005,0.005,0.05,0.5,1),limits = c(1e-4,1)) + theme_bw() + theme(legend.position="top") +
  ylab("P-value") + xlab("Sample size") + guides(fill=guide_legend(title="Method: ")))
dev.off()


typeone$reject = typeone$P < 0.05
print(ftable(reject ~ method + n, data = typeone))


# Simulations for evaluating asymptotic convergence in distribution

Thetavals_2 = data.table("n" = rep(0,J*length(N)), "Theta" = rep(0,J*length(N)))

j = 1
for(n in c(50,60,70)){
  i = 1
  while(i< J*3){
    indices = sample(1:136,n)
    s_data = subset(surv_data,population==1)[indices,]
    s_data$population[1:(n/2)] = 2
    tryCatch({
      out = Theta_hat(s_data)
      Thetavals_2[j] = list( "n" =n , "Theta" = as.numeric(out))
      j = j + 1
      i = i + 3
    },error = function(err) print(err))
    
  }
}

# compute the theoretical Gaussian densities

n = 50
indices = sample(1:136,n)
s_data = subset(surv_data,population==1)[indices,]
s_data$population[1:(n/2)] = 2
out_50 = Theta_hat(s_data)

n = 60
indices = sample(1:136,n)
s_data = subset(surv_data,population==1)[indices,]
s_data$population[1:(n/2)] = 2
out_60 = Theta_hat(s_data)

n = 70
indices = sample(1:136,n)
s_data = subset(surv_data,population==1)[indices,]
s_data$population[1:(n/2)] = 2
out_70 = Theta_hat(s_data)

grid <- with(dd, seq(-750,750, length = 200))
normaldens =
  data.frame( 
    Theta = rep(grid,3),
    n = rep(c(50,60,70),each=length(grid)),
    density = c(dnorm(grid, mean = 0, sd = sqrt(slot(out_50,'variance'))),
                dnorm(grid, 0, sqrt(slot(out_60,'variance'))),
                dnorm(grid, 0, sqrt(slot(out_70,'variance'))))
  )


pdf("fig7.pdf",family="CM Roman", width=8, height=3)
pp2 = ggplot(Thetavals_2,aes(Theta)) + geom_histogram(aes(x = Theta, y = ..density..),bins = 80,fill = "grey") + facet_wrap(~n) +
  theme_bw() + xlab(expression(hat(Theta)))+ 
  geom_line(aes(y = density), data = normaldens, colour = "red", size=1.1)  
print(pp2)
dev.off()
\end{lstlisting}

\end{document}